\documentclass[a4paper,UKenglish]{article}

\usepackage{microtype}
\usepackage{fullpage}
\usepackage{esvect}

\MakeRobust{\vv}
\usepackage{graphics}
\usepackage{epstopdf}
\newcommand{\leaf}{leaf}
\usepackage{authblk}

\usepackage{graphicx}
\usepackage{amsthm}
\usepackage{amsmath}
\usepackage{enumerate}
\usepackage{color}
\usepackage{xspace}
\usepackage{url}

\usepackage{amsfonts}\usepackage{amssymb}
\usepackage{thmtools}\usepackage{mathscinet}
\usepackage{thm-restate}
    \newtheorem{theorem}{Theorem}
    \newtheorem{corollary}[theorem]{Corollary}
    \newtheorem{lemma}[theorem]{Lemma}
    \newtheorem{definition}[theorem]{Definition}

    	\definecolor{darkgreen}{rgb}{0.01, 0.93, 0.29}
\definecolor{lightbrown}{rgb}{0.91, 0.4, 0.11}
\usepackage{framed}

\title{Oriented Diameter of Planar Triangulations}

\author[1]{Debajyoti Mondal}
\author[2]{N. Parthiban}
\author[3]{Indra Rajasingh}
 
\affil[1]{Department of Computer Science, University of Saskatchewan, Saskatoon, Canada\\
  \texttt{dmondal@cs.usask.ca}} 
\affil[2]{School of Computing, SRM Institute of Science and Technology, Chennai, India\\
  \texttt{parthiban24589@gmail.com}}
\affil[3]{School of Advanced Sciences, Vellore Institute of Science and Technology, Chennai, India\\
  \texttt{indrarajasingh@yahoo.com}}

\usepackage[textsize=tiny]{todonotes}
\usepackage{verbatim}

\newcommand{\jyoti}[1]{{\color{black} #1}}
\newcommand{\parthiban}[1]{{\color{black} #1}}

\begin{document}
\maketitle              
\begin{abstract}
The diameter of an undirected or a directed graph is defined to be the maximum shortest path distance over all pairs of vertices in the graph. Given an undirected graph $G$, we examine the problem of assigning directions to each edge of $G$ such that the diameter of the resulting oriented graph is minimized. The minimum  diameter over all  strongly connected orientations is called the oriented diameter of $G$. The problem of determining the oriented diameter of a graph is known to be NP-hard, but the time-complexity question is open for planar graphs. In this paper we compute the exact value of the oriented diameter for triangular grid graphs. We then prove an $n/3$ lower  bound and an  $n/2+O(\sqrt{n})$ upper bound on the oriented diameter of planar triangulations.  It is known that given a planar graph $G$ with bounded treewidth and a fixed positive integer $k$, one can determine in linear time whether the oriented diameter of $G$ is at most $k$. In contrast, we consider a weighted version of the oriented diameter problem and show it to be is weakly NP-complete for planar graphs with bounded pathwidth.
\end{abstract}


\section{Introduction}

\jyoti{An undirected graph is called \emph{oriented} when each edge of the graph is assigned an orientation. Computing such orientations often \parthiban{requires} the resulting directed graph to be \emph{strongly connected}, i.e., every vertex in the directed graph must be reachable from every other vertex. This is useful in  transforming two-way traffic or communication networks to one-way networks especially when one-way communication is preferred or more cost effective over two-way communication channels~\cite{Kumar2020,DBLP:journals/networks/RobertsX92}, as well as finds application in the context of network broadcasting and gossiping~\cite{DBLP:journals/dam/FraigniaudL94,DBLP:journals/siamcomp/Krumme92}.}  

The \emph{diameter} of a directed or undirected graph is the maximum shortest path distance over all pairs of vertices, where the distance of a path is measured by the number of its edges. The oriented diameter $OD(G)$ of an undirected graph $G$ is the \jyoti{smallest} diameter over all the  strongly connected orientations of $G$. 
\jyoti{In 1978, Chavatal et al.~\cite{Chvatal1978} proved that determining whether the oriented diameter of a graph is at most two is NP-complete. They showed that the oriented diameter of a $2$-edge-connected graph with diameter $2$ is at most $6$, and there exist graphs achieving this upper bound. For graphs with diameter 3, the known upper and lower bounds are 9 and 11, respectively~\cite{Kwok2010}. Several studies attempted to provide good upper bounds on the orientated diameter problem of connected and bridgeless graphs~\cite{Fomin2001,Sascha2012,Dankelmann2018}.} In 2001,  Fomin et al.~\cite{Fomin2001} discovered the relation $OD(G)  \leq 9 \gamma (G) - 5$ between the oriented diameter and the size $ \gamma(G)$ of a minimum dominating set.  Dankelmann et. al.~\cite{Dankelmann2018} showed that every bridgeless graph $G$ of order $n$ and maximum degree $\Delta$ has an orientated diameter at most $n-\Delta-3$.

\jyoti{A rich body of literature examined  oriented diameter for} interconnection networks~\cite{Fomin2004, koh1996, koh2002, Ng2005} \jyoti{and for various  interesting graph classes}. Gutin et al.~\cite{Gutin2002}  studied some well-known classes of strong digraphs where oriented diameter exceeds the diameter of the underlying undirected graph only by a small constant.  Fujita et al.~\cite{Fujita2013} considered the problem of finding the minimum oriented diameter of star graphs and proved an upper bound of $[5n/2]+2$ for any $n \geq 3$, which is a significant improvement  over the upper bound $2n(n - 1)$ derived by Chvatal and Thomassen~\cite{Chvatal1978}. \jyoti{Fomin et al.~\cite{Fomin2002} showed that computing oriented diameter remains NP-hard for split graphs and provided approximation algorithms for chordal graphs. Later, they showed that the oriented diameter of an  AT-free graph is upper bounded by a linear function in its graph diameter~\cite{Fomin2004}.} 

\jyoti{
We consider oriented diameter of planar graphs. Eggemann  and Noble~\cite{Eggemann2009} showed that given a planar graph $G$ with bounded treewidth and a fixed positive integer $k$, one can determine in linear time whether the oriented diameter of $G$ is at most $k$. They also showed how to remove the dependency on the treewidth and gave an algorithm that given a fixed positive integer $k$,   can decide whether the oriented diameter of a planar graph is at most $k$ in linear time.   
Recently, Wang et al.~\cite{DBLP:journals/jgt/WangCDGSV21} have showed that the diameter of a  maximal outerplanar graph with at least three vertices is upper bounded by   $\lceil n/2 \rceil$ with four exceptions, and the upper bound is sharp. 

\subsection*{Our Contribution}
In this paper we compute the exact  value on the oriented diameter for  triangular grid graphs.  
This result  relates to the vast literature that attempts to compute diameter preserving or optimal orientation for well-known graph classes  (e.g., for two-dimensional  torus~\cite{DBLP:journals/networks/KonigKL98},  two-dimensional  grid~\cite{DBLP:journals/siamdm/RobertsX88,DBLP:journals/networks/RobertsX92}, hypercube~\cite{McCanna}, products of graphs~\cite{DBLP:journals/dm/KohT00a}).  We then generalize the idea of computing oriented diameter of triangular grid graphs to give an algorithm for planar triangulations. Given a planar triangulation with $n$ vertices, we show how to orient its edges such that the diameter of the resulting oriented  graph is upper bounded by $n/2+O(\sqrt{n})$. This is interesting since $\lceil n/2 \rceil$  is already known to be an upper bound on the oriented diameter for the maximal outerplanar graphs~\cite{DBLP:journals/jgt/WangCDGSV21}.   
 We next show that there exist planar triangulations with oriented diameter at least $n/3$. Finally, we show that the weighted version of the oriented diameter problem is weakly NP-complete even for planar graphs of bounded pathwidth, which  contrasts the linear-time algorithm of Eggemann  and Noble~\cite{Eggemann2009} for the unweighted variant. 
}

\section{Oriented Diameter of a Triangular Grid}
In this section we compute the exact value of the oriented diameter of triangular grid graphs. A triangular tesselation of the plane with equilateral triangles is called a triangular sheet or triangular grid. The vertices are the intersection of lines and the lines between two vertices are the edges. 

\begin{definition}
Consider the ordered 3-tuples of integers $(i, j, k)$ such that $i+j+k=r$. Let these 3-tuples represent the  vertices and let two vertices be joined if the sum of the absolute differences of their coordinates is $2$. The graph   generated is referred to as a triangular grid graph $T_r$ of dimension $r$. 
\end{definition}

\begin{figure}[h]
\centering
\includegraphics[scale=0.4]{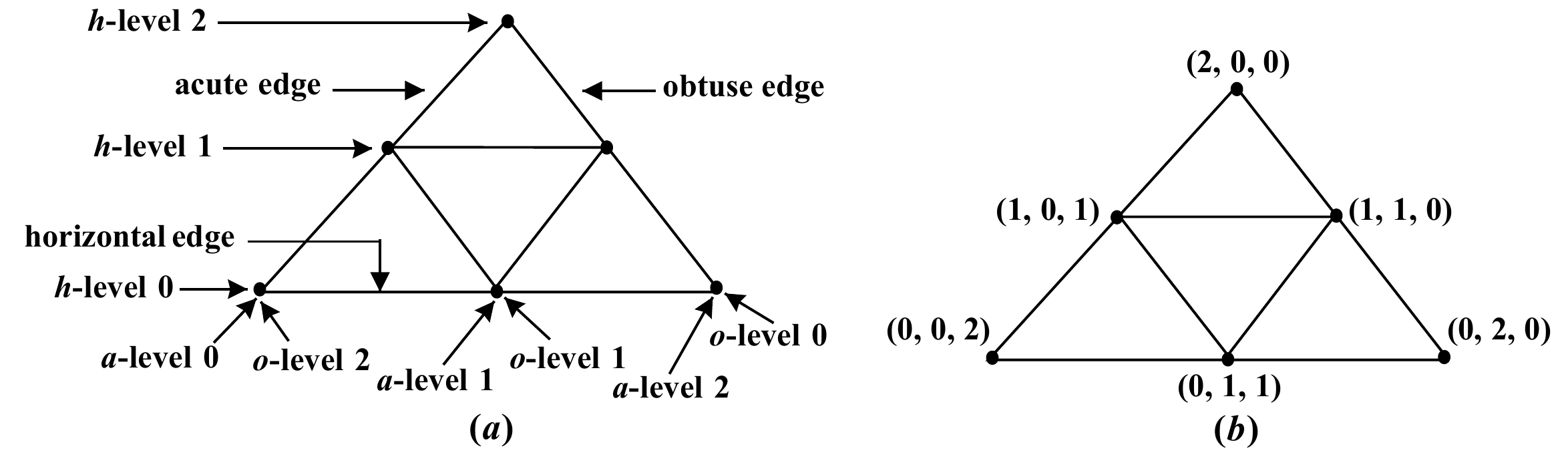}
\includegraphics[scale=0.37]{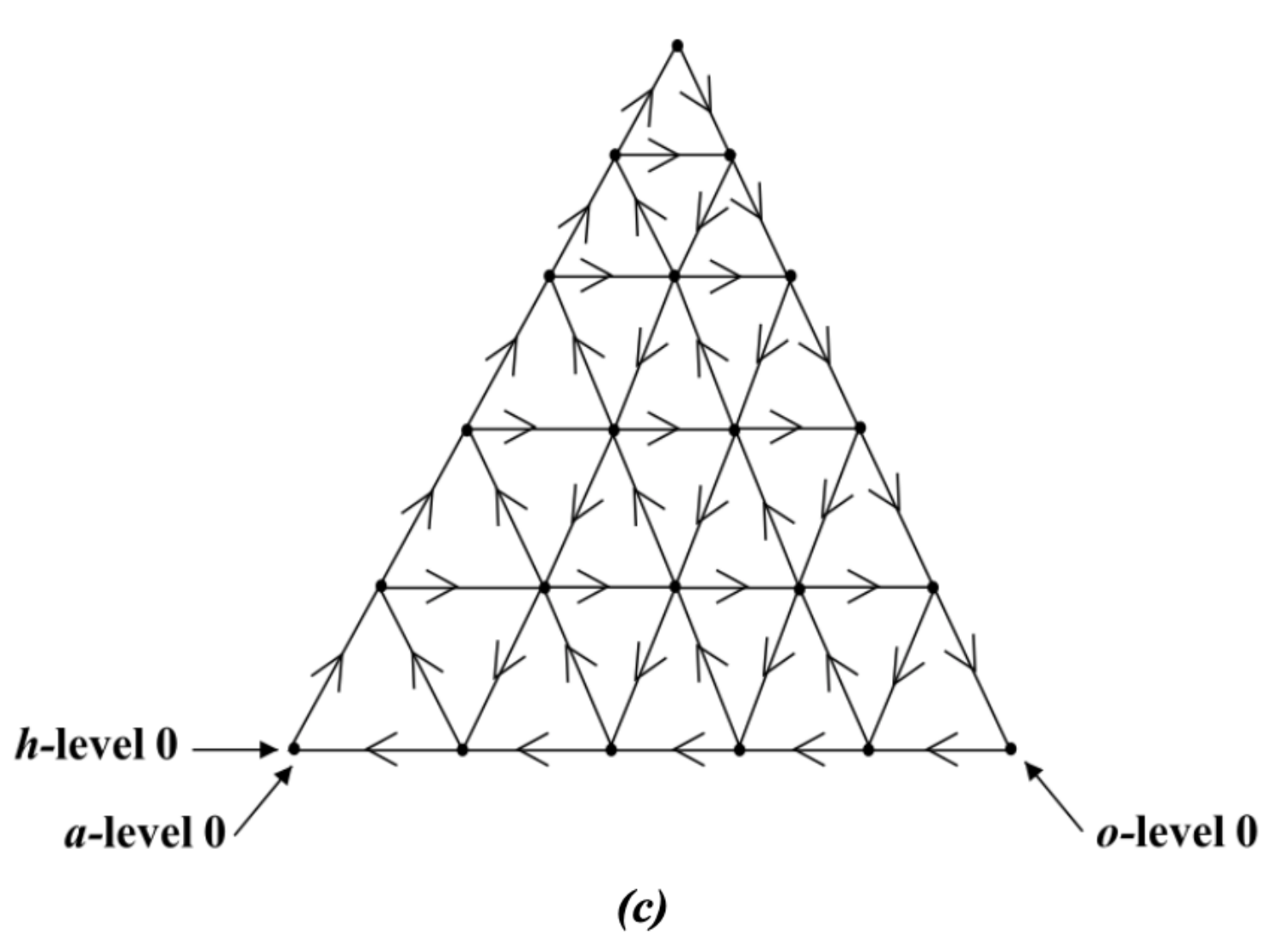}
\caption{The triangular-grid network $T_{2}$.  Illustration for (a) levels and (b) 3-tuples.  \parthiban{(c)} An oriented triangular grid network $T(5)$. }
\label{t5}
\end{figure}

There are $r+1$ levels in a triangular grid $T_r$ and in each level $i$, there are $i+1$ vertices, $0 \leq i \leq r$. This implies that $T_r$ has $(r+1)(r+2)/2$ vertices and $3r(r+1)/2$ edges. Diameter of $T_r$ is $r$. 
The edges of $T_{r}$ can be partitioned into horizontal edges, acute edges, and obtuse edges. The shortest path comprising of horizontal (acute, obtuse) edges with end vertices of degree $2$ in $T_r$ is said to be at $h$-level ($a$-level, $o$-level) ~$0$. Inductively, the path through the parents of vertices at level $i$ is said to be at $h$-level~ ($a$-level,~ $o$-level)~ $i + 1$,~ $0 \leq i \leq r-1 $. The vertex $v = (i, j, k)$ in $T_{r}$ is the point of intersection of the paths representing its $h$-level $i$, $a$-level $j$, and $o$-level $k$. See Figures $1(a)$ and $1(b)$. 


A simple observation on the length of a shortest cycle passing through any two vertices of length 2 in $T_r$, $r \geq 2$ yields the following result. 
	
\begin{lemma} 
\label{r+1}
Let $G$ be the triangular grid $T_r$, $r \geq 2$. Then $OD(G) \geq r+1$.
\end{lemma}
\begin{proof}
The diameter of $T_r$ is $r$, which is realized between two end vertices of $h$-$level~0$ line. Any shortest cycle in $G$ passing through these two end vertices has length $r + (r+1) = 2r + 1$. Hence $OD(G) \geq r+1$.
\end{proof}


The following algorithm yields an oriented triangular grid with diameter $r+1$. In the sequel, if $\vv {P}$ and $\vv {Q}$ are directed paths then $\vv {P} \circ \vv {Q}$ denotes the concatenation of the paths $\vv {P}$ and $\vv {Q}$ with end of $\vv {P}$ as the beginning of $\vv {Q}$.

\bigskip

\noindent
\textbf{Input:} A triangular grid graph $T_r$ of dimension $r$.

\noindent
\textbf{Output:} An orientation of $T_r$  with diameter $r+1$.

\noindent
\textbf{Algorithm:} Direct the $h$-$level~ 0$ line from right to left; $a$-$level ~0$ line from bottom to top; $o$-$level ~0$ 
line from top to bottom. Direct all other horizontal lines from left to right, acute lines from top to bottom and the obtuse lines from bottom to top. See Figure $\ref{t5}$\parthiban{(c)}.

 \bigskip\noindent
{\bf Proof of correctness:} Given any two vertices $u$ and $v$ in $T_r$, we have to show that there exist a $(u, v)$-directed path and a $(v, u)$-directed path, both of length at most $r+1$. Let $u$ be $(i,j,k)$ and $v$ be $(l,m,n)$. Without loss of generality assume that $i<l$.

\begin{figure}[h]
\centering
\includegraphics[scale=0.37]{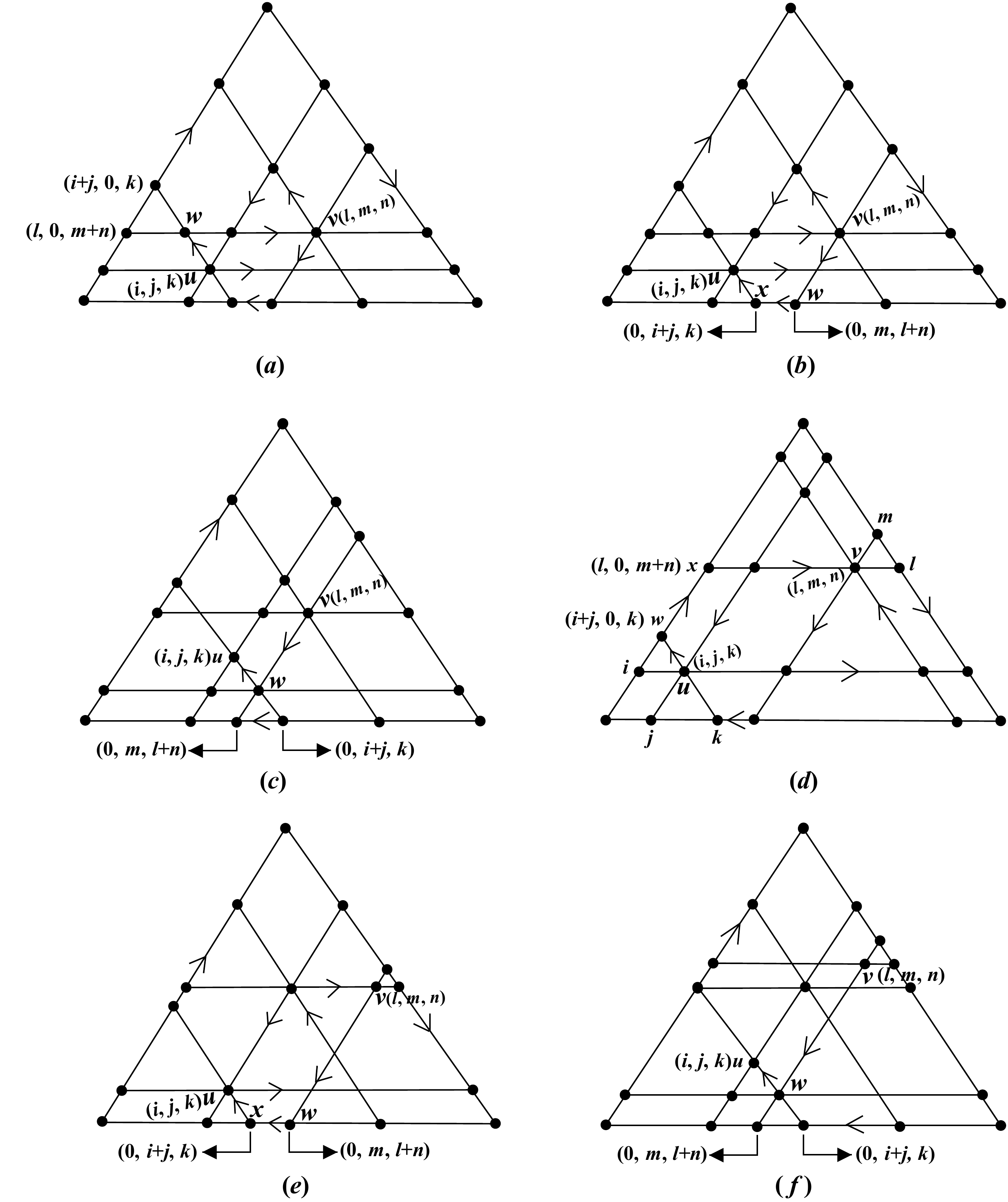}
\caption{(a) $\vv{P}$: $(u,w)$-directed path,  $\vv{Q}$:$(w,v)$-directed path,  $\vv{P} \circ \vv{Q}$ is directed path from $u$ to $v$,
 (b) $\vv {P}:(v,w)$-directed path, $\vv {Q}:(w,x)$-directed path, $\vv {R}:(x,u)$-directed path, $\vv {P} \circ \vv {Q} \circ \vv {R}$ is directed path from $v$ to $u$,
 (c) $\vv {P}:(v,w)$-directed path, $\vv {Q}:(w,u)$-directed path, $\vv {P} \circ \vv {Q} $ is directed path from $v$ to $u$,
 (d) $\vv {P}:(u,w)$-directed path, $\vv {Q}$:$(w,x)$-directed path, $\vv {R}:(x,v)$-directed path, $\vv {P} \circ \vv {Q} \circ \vv {R}$ is $(u,v)$-directed path,
 (e) $\vv {P}:(v,w)$-directed path, $\vv {Q}:(w,x)$-directed path, $\vv {R}:(x,u)$-directed path, $\vv {P} \circ \vv {Q} \circ \vv {R}$ is directed path from $v$ to $u$ and
 (f) $\vv {P}:(v,w)$-directed path, $\vv {Q}:(w,u)$-directed path, $\vv {P} \circ \vv{Q} $ is directed path from $v$ to $u$.}
\label{fig_3}
\end{figure}

\noindent \textbf{Case 1(\textit{a})}: $j<m$ and $i+j \geq l$

\noindent\textbf{Path from $u$ to $v$:}
Let $\vv {P}$ be the directed path from $u$ along the $o$-level $k$ till it reaches vertex $w$ in the $h$-level $l$. Let $\vv {Q}$ be the directed path from $w$ along the $h$-level $l$ till it reaches vertex $v$. Then $\vv {P} \circ \vv {Q}$ is a directed path from $u$ to $v$ of length $(l-i)+(k-n) \leq j+k-n = r-(i+n) \leq r$. See Figure $\ref{fig_3}(a)$ .

\noindent\textbf{Path from $v$ to $u$:}
When $(i+j)<m$, let $\vv {P}$ be the directed path from $v$ along the $a$-level $m$ till it reaches vertex $w$ in $h$-level $0$. Let $\vv {Q}$ be the directed path from $w$ along $h$-level $0$ till it reaches $x$ at $o$-level $k$. Let $\vv {R}$ be the directed path from $x$ along $o$-level $k$ till it reaches vertex $u$. 
Then $\vv {P} \circ \vv {Q} \circ \vv {R}$ is a directed path from $v$ to $u$ of length 
$l+(m-(i+j))+ i = r-n-j < r$. See Figure $\ref{fig_3}(b)$.
When $m \leq i+j$, let $\vv {P}$ be the directed path from $v$ along the $a$-level $m$ till it reaches vertex $w$ in $o$-level $k$; let $\vv {Q}$ be the directed path from $w$ along $o$-level $k$ till it reaches $u$. 
Then $\vv {P} \circ \vv {Q} $ is a directed path from $v$ to $u$ of length at most of
$l+(m-j) = r-n-j < r$. See Figure $\ref{fig_3}(c)$. \newline\textbf{Case 1(\textit{b})}: $j<m$ and $i+j < l$

\noindent\textbf{Path from $u$ to $v$:}
Let $\vv {P}$ be the directed path from $u$ along the $o$-level $k$ till it reaches vertex $w$ on the $a$-level $0$. 
Let $\vv {Q}$ be the directed path from $w$ along the $a$-level $0$ till it reaches vertex $x$ at $h$-level $l$. 
Let $\vv {R}$ be the directed path from $x$ along $h$-level $l$ till it reaches $v$. 
Then $\vv {P} \circ \vv {Q} \circ \vv {R}$ is a directed path from $u$ to $v$ of length $j+(l-(i+j))+m = l+m-i = r-i-n <r$. See Figure $\ref{fig_3}(d)$. \newline \noindent \textbf{Path from $v$ to $u$:}
When $(i+j)<m$, let $\vv {P}$ be the directed path from $v$ along the $a$-level $m$ till it reaches vertex $w$ in $h$-level $0$. Let $\vv {Q}$ be the directed path from $w$ along $h$-level $0$ till it reaches $x$ at $o$-level $k$. Let $\vv {R}$ be the directed path from $x$ along $o$-level $k$ till it reaches vertex $u$. 
Then $\vv {P} \circ \vv {Q} \circ \vv {R}$ is a directed path from $v$ to $u$ of length 
$l+(m-(i+j))+ i = r-n-j < r$. See Figure $\ref{fig_3}(e)$.
When $m \leq i+j$, let $\vv {P}$ be the directed path from $v$ along the $a$-level $m$ till it reaches vertex $w$ in $o$-level $k$; let $\vv {Q}$ be the directed path from $w$ along $o$-level $k$ till it reaches $u$. 
Then $\vv {P} \circ \vv {Q} $ is a directed path from $v$ to $u$ of length at most of
$l+(m-j) = r-n-j < r$. See Figure $\ref{fig_3}(f)$.

\begin{figure}[h]
\centering
\includegraphics[scale=0.36]{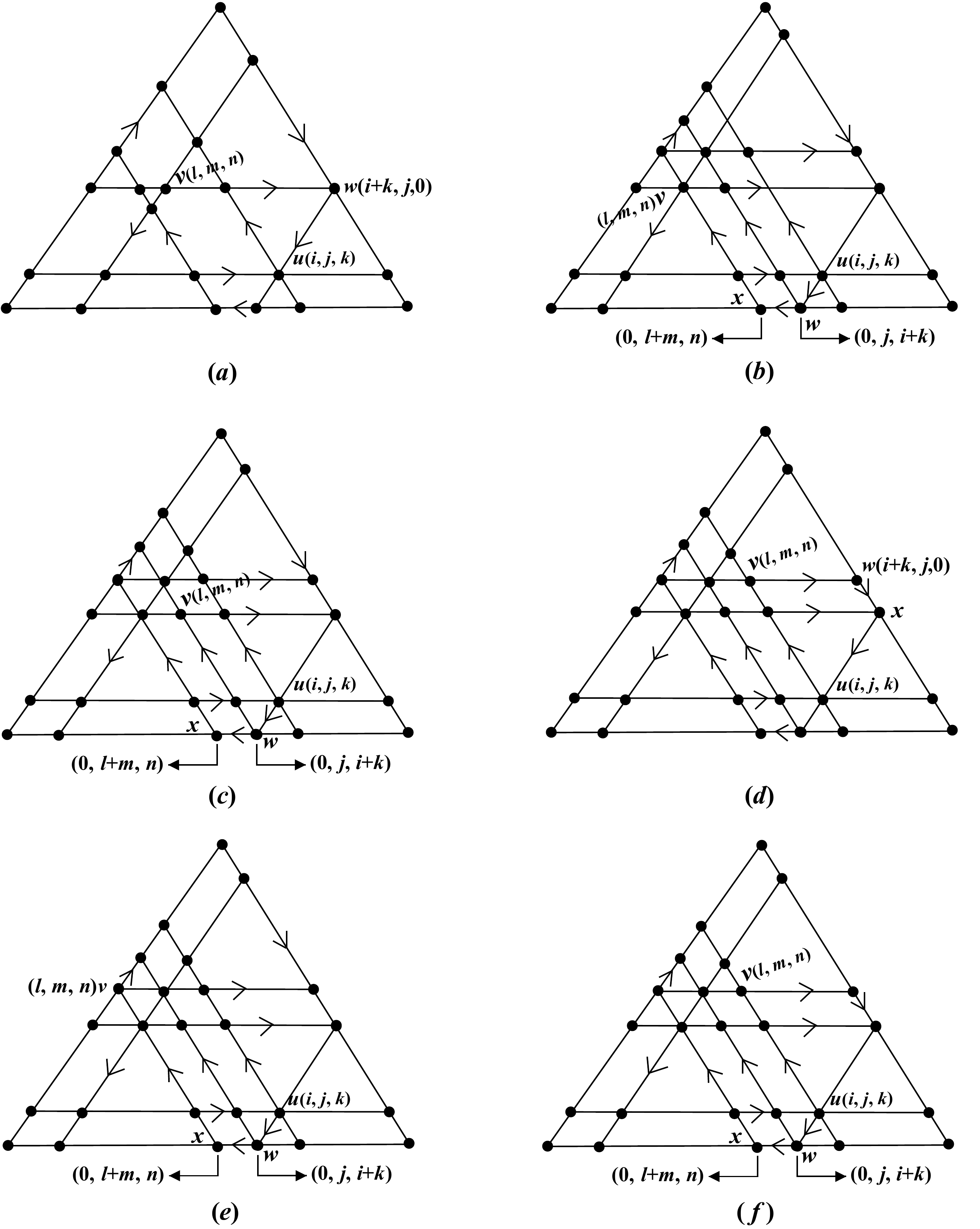}
\caption{$(a)$ $\vv {P}:(v,w)$-directed path, $\vv {Q}$:$(w,u)$-directed path,  $\vv {P} \circ \vv {Q}$ is directed path from $v$ to $u$,
$(b)$ $\vv {P}:(u,w)$-directed path, $\vv {Q}:(w,x)$-directed path, $\vv {R}:(x,v)$-directed path, $\vv {P} \circ \vv {Q} \circ \vv {R}$ is directed path from $u$ to $v$,
$(c)$ $\vv {P}:(u,w)$-directed path, $\vv {Q}:(w,v)$-directed path, $\vv {P} \circ \vv {Q}$ is $(u,v)$-directed path,
$(d)$ $\vv {P}:(v,w)$-directed path, $\vv {Q}$:$(w,x)$-directed path, $\vv {R}:(x,u)$-directed path, $\vv {P} \circ \vv {Q} \circ \vv {R}$ is $(v,u)$-directed path,
$(e)$ $\vv {P}:(u,w)$-directed path, $\vv {Q}:(w,x)$-directed path, $\vv {R}:(x,v)$-directed path, $\vv {P} \circ \vv {Q} \circ \vv {R}$ is directed path from $u$ to $v$,
$(f)$ $\vv {P}:(u,w)$-directed path, $\vv {Q}:(w,v)$-directed path, $\vv {P} \circ \vv {Q}$ is $(u,v)$-directed path.
}
\label{fig_4_case_2}
\end{figure}

\noindent \textbf{Case 2(\textit{a})} : $j \geq m$ and $i+k \geq l$. 

\noindent \textbf{Path from $v$ to $u$:}
Let $\vv {P}$ be the directed path from $v$ along the $h$-level $l$ till it reaches vertex $w$ in the $a$-level $j$. Let $\vv {Q}$ be the directed path from $w$ along the $a$-level $j$ till it reaches vertex $u$. Then $\vv {P} \circ \vv {Q}$ is a directed path from $v$ to $u$ of length $(l-i)+(j-m) \leq k+j-m = r-(i+m) \leq r$. 
See Figure $\ref{fig_4_case_2}(a)$ .

\noindent\textbf{Path from $u$ to $v$:}
When $i+k \geq m$, let $\vv {P}$ be the directed path from $u$ along the $a$-level $j$ till it reaches vertex $w$ in $h$-level $0$. Let $\vv {Q}$ be the directed path from $w$ along $h$-level $0$ till it reaches $x$ at $o$-level $n$. Let $\vv {R}$ be the directed path from $x$ along $o$-level $n$ till it reaches vertex $v$. 
Then $\vv {P} \circ \vv {Q} \circ \vv {R}$ is a directed path from $v$ to $u$ of length 
$i+(j-(l+m))+ l = r-(k+m) < r$. 
See Figure $\ref{fig_4_case_2}(b)$.
When $i+k < n$, let $\vv {P}$ be the directed path from $u$ along the $a$-level $j$ till it reaches vertex $w$ in $h$-level $0$; let $\vv {Q}$ be the directed path from $w$ along $o$-level $n$ till it reaches $v$.
Then $\vv {P} \circ \vv {Q} $ is a directed path from $v$ to $u$ of length at most of
$i+(j-m) = r-(k+m) < r$. 
See Figure $\ref{fig_4_case_2}(c)$.

\noindent\textbf{\\ Case 2(\textit{b})} : $j \geq m$ and $i+k < l$.

\noindent\textbf{Path from $v$ to $u$:}
Let $\vv {P}$ be the directed path from $v$ along the $h$-level $l$ till it reaches vertex $w$ in the $o$-level $0$. Let $\vv {Q}$ be the directed path from $w$ along the $a$-level $j$ till it reaches vertex $x$ in the $a$-level j.
Let $\vv {R}$ be the directed path from $x$ along $a$-level $j$ till it reaches vertex $u$. 
Then $\vv {P} \circ \vv {Q} \circ \vv {R}$ is a directed path from $v$ to $u$ of length 
$n+(l-(i+k))+ k = r-(m+i) < r$. 
See Figure $\ref{fig_4_case_2}(d)$ .

\noindent\textbf{Path from $u$ to $v$:}
When $i+k \geq m$, let $\vv {P}$ be the directed path from $u$ along the $a$-level $j$ till it reaches vertex $w$ in $h$-level $0$. Let $\vv {Q}$ be the directed path from $w$ along $h$-level $0$ till it reaches $x$ at $o$-level $n$. Let $\vv {R}$ be the directed path from $x$ along $o$-level $n$ till it reaches vertex $v$. 
Then $\vv {P} \circ \vv {Q} \circ \vv {R}$ is a directed path from $v$ to $u$ of length 
$i+(j-(m+l))+ l = r-(k+m) < r$. 
See Figure $\ref{fig_4_case_2}(e)$.
When $i+k < n$, let $\vv {P}$ be the directed path from $u$ along the $a$-level $j$ till it reaches vertex $w$ in $h$-level $0$; let $\vv {Q}$ be the directed path from $w$ along $o$-level $n$ till it reaches $v$.
Then $\vv {P} \circ \vv {Q} $ is a directed path from $v$ to $u$ of length at most of
$i+(j-m) = r-(k+m) < r$. 
See Figure $\ref{fig_4_case_2}(f)$.

\noindent\textbf{\\Case 3}: Both the vertices $u$ and $v$ lie in the same $h$-level, $a$-level or $o$-level. As long as $u$ and $v$ are not both vertices of degree 2 in $T_r$, arguments similar to that of Case 1 show that $\vv {d}(u, v)$ and $\vv {d}(v, u)$ are at most $r$. Suppose that both $u$ and $v$ are  of degree 2 in $T_r$, then if $\vv{d}(u, v) = r$, then $\vv{d}(v, u) = r+1$. Thus the oriented graph $T_r$ has diameter $r+1$. 

\bigskip
\noindent
The above analysis of the algorithm for triangular grid graphs  and Lemma \ref{r+1} yield the following result.

\begin{theorem}
Let $G$ be the triangular $T_{r}^{}, r>2$. Then $OD(G)=r+1$.
\end{theorem}

\jyoti{
\section{Upper Bound on the Oriented Diameter of a Planar Triangulation}
\label{sec:pt}

In this section we provide an $n/2+O(\sqrt{n})$  upper bound on the oriented diameter of a planar triangulation. The idea is to use edge orientations similar to the one we used for the triangular grid, i.e., three outgoing edges at each internal vertex with interleaving incoming edges  and a clockwise orientation on the outer face. In fact, every planar triangulation is known to admit such an orientation for its internal vertices, which is obtained by via a `Schnyder realizer'. Since we use this concept extensively, we first briefly review some preliminary definitions and notation related to planar graphs and Schnyder realizer.

\subsection{Planar Graphs and  Schnyder Realizer}
\label{sec:pre}


A \emph{planar graph} is a graph that can be drawn on the Euclidean plane such that no two edges cross except possibly at their common endpoint. A \emph{plane graph} $G$ is a planar graph with a fixed planar embedding on the plane.  
 $G$ delimits the plane into connected regions called faces.
 The unbounded face is the outer face of $G$ and all the other faces are the inner faces of $G$.
 $G$ is called \emph{triangulated} 
 if every face (including the outerface) 
 of $G$ contains exactly three vertices on its boundary. The vertices on the outer face of $G$ are called the \emph{outer vertices} and
 all the remaining vertices are called the \emph{inner vertices}. The edges on the outer face are called the \emph{outer edges} of $G$.

Let $G=(V,E)$ be a triangulated plane graph with the outer vertices $v_m,v_r$ and $v_l$ in clockwise order on the outer face. Then the internal edges of $G$  can be directed in a way such that every inner vertex has three directed  paths which are vertex disjoint (except that they start at at the common vertex $v$) and end at the three outer vertices~\cite{S90}. In other words, there are three directed trees $T_m$, $T_r$ and $T_l$ rooted at $v_m,v_r$ and $v_l$, respectively, that span all the internal vertices of $G$. Let \parthiban{us} denote the edges of  $T_m,T_r$ and $T_l$ as the $m$-edges, $r$-edges and $l$-edges, respectively. Then the edges at every internal vertex $v$ of $G$ are directed in a certain order, as shown in Figure~\ref{figure:cann}(c). The trees $T_m,T_r$ and $T_l$ are referred to as \emph{Schnyder realizer}.  Figure~\ref{figure:cann}(a) illustrates a Schnyder realizer and Figure~\ref{figure:cann}(b) shows the tree $T_l$ separately.}

\begin{figure}[h]
\centering
\includegraphics[width=\textwidth]{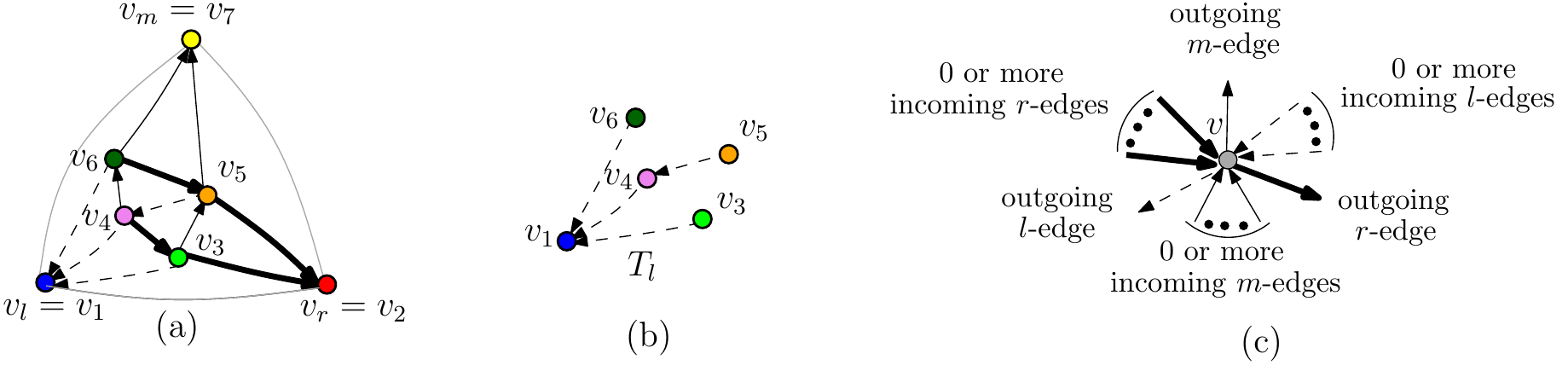}
\caption{(a) A Schnyder realizer. (b) Illustration for $T_l$. (c) Illustration for the edge orientations around a vertex. } 
\label{figure:cann} 
\end{figure}


There exists a Schnyder realizer $T_l,T_r,T_m$ of $G$ with $\leaf(T_l)+\leaf(T_r)+\leaf(T_m)= 2n-5-\delta_0$
\cite{BonichonSM02}, where $0\le \delta_0 \le \lfloor (n-1)/2 \rfloor$ is the number of cyclic faces.
 It is known that there exists a Schnyder realizer where one tree  has at least $\lceil$
 $ {(n+1)}/{2}\rceil$ leaves and such a realizer 
 can be computed in linear time \cite{ZhangH05}.

 \jyoti{
\subsection{An Initial Upper Bound}
\label{early}
We first give an algorithm to compute a $\lceil {3(n-1)}/{4} \rceil + 2$ upper bound on oriented diameter of $G$, and in the subsequent section we improve the bound to  $n/2+O(\sqrt{n})$. 

\bigskip
\noindent
\textbf{Algorithm:} We first compute a Schnyder realizer $T_l,T_r,T_m$  such that one tree (assume without loss of generality that $T_m$) has at least $\lceil \frac{(n+1)}{2}\rceil$ leaves~\cite{ZhangH05}.  
We now  assign the edge orientations as follows (Figure~\ref{figure:algo}(a)).

\bigskip
\noindent
Orient the outer edges of $G$ in clockwise order:   $\vv{(v_m,v_r)}$, $\vv{(v_r,v_l)}$ and $\vv{(v_l,v_m)}$. 

\noindent
Orient the $m$-edges such that each vertex of $T_m$ points at  its ancestor in $T_m$.

\noindent
Orient the $l$-edges such that each vertex of $T_l$ points at  its  descendent in $T_l$.

\noindent
Orient the $r$-edges such that each vertex of $T_r$ points at  its  descendent in $T_r$.
\bigskip

\begin{figure}[h]
\centering
\includegraphics[width=\textwidth]{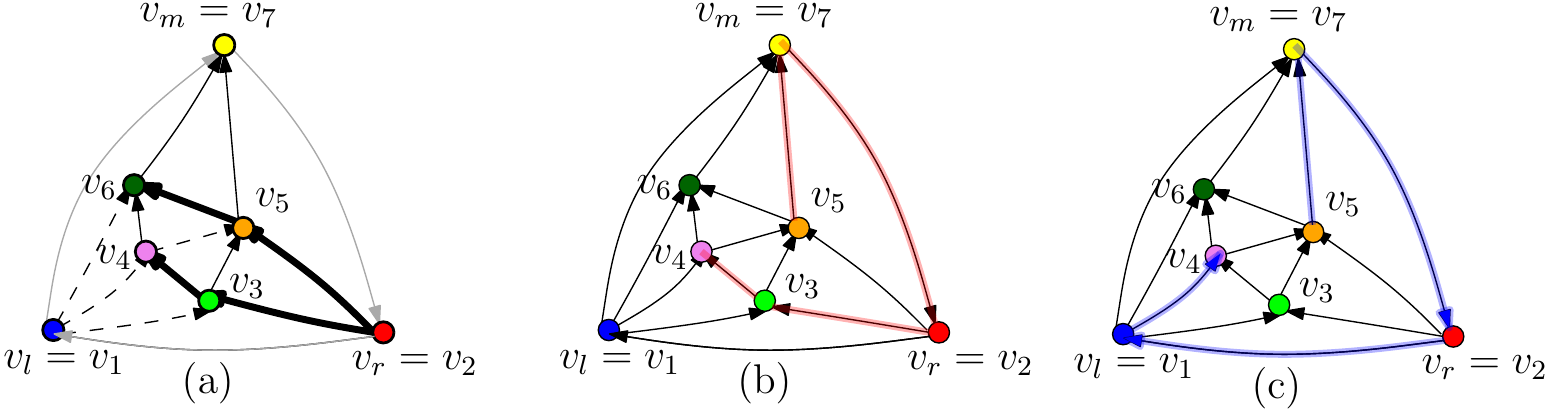}
\caption{(a) Illustration for the algorithm of Section~\ref{early}. (b)--(c) The paths  $\vv{Q}$ and $\vv{Q'}$, which are shown in red and blue, respectively. Here $v=v_5$ and $w = v_4$. } 
\label{figure:algo} 
\end{figure}

\noindent
\textbf{Proof of correctness:} We now show that for every pair of vertices $u,w$ in $G$, their shortest path length is bounded by at most $\lceil {3(n-1)}/{4} \rceil + 2$.

\bigskip
\noindent
\textbf{Case 1 (both $u,w$ are internal vertices of $G$)}:  Let $\vv{P}_{u,v_m}$ be the path that starts at $u$ and follows the $m$-edges to reach $v_m$. Let $\vv{P}_{v_r,w}$  be the path that starts at $v_r$  and follows the $r$-edges to reach $w$. Similarly, let  $\vv{P}_{v_l,w}$ be the path that starts at $v_l$ and follows the $l$-edges to reach $w$.  

Consider now  two  $u$ to $w$ paths (Figure~\ref{figure:algo}(b)--(c)). One is the path $\vv{Q}$ that  first travels along $\vv{P}_{u,v_m}$, then along the edge $\vv{(v_m,v_r)}$ and finally, along $\vv{P}_{v_r,w}$. The other is the path $\vv{Q'}$ that  first travels along $\vv{P}_{u,v_m}$, then along the edges $\vv{(v_m,v_r)}$ and $\vv{(v_r,v_l)}$, and finally, along $\vv{P}_{v_l,w}$. We now show \parthiban{that} at least one of these two paths are of length at most  $\lceil {3(n-1)}/{4} \rceil + 2$.

Since $T_m$ has at least $\lceil {(n+1)}/{2}\rceil$ leaves, the length $\beta$ of $\vv{P}_{u,v_m}$ is at most $\lceil n-\frac{(n+1)}{2} \rceil =  \lceil {(n-1)}/{2} \rceil$. Since the paths  $\vv{P}_{v_r,w}$ and $\vv{P}_{v_l,w}$ are vertex disjoint (except that they start at $w$), the sum of their lengths is at most $(n-\beta-1)$. Here the $-1$ term represents that we can safely skip the vertex $v_m$.  We now can assume without loss of generality that  $\vv{P}_{v_l,w}$ has at most  $\lceil (n-\beta-1)/2\rceil$ edges. Therefore, the length of $\vv{Q'}$ is bounded by at most 
\begin{align*}
    &\beta + \lceil (n-\beta-1)/2\rceil+2 \\
    &= \lceil {(n-1)}/{2} + \beta/2\rceil   + 2, \text{ where  } \beta\le \lceil {(n-1)}/{2} \rceil \\
    &=  \lceil {3(n-1)}/{4} \rceil + 2
\end{align*}

Similarly, we can find a $w$ to $u$ path of length  at most $ \lceil {3(n-1)}/{4} \rceil + 2$ by swapping the \parthiban{roles} of $u$ and $w$ in the argument above.

\bigskip
\noindent
\textbf{Case 2 (at least one of $u,w$ is an outer vertex of $G$)}: If both $u,w$ are outer vertices then they can reach each other following at most two outer edges. If exactly one is an outer vertex, then without loss of generality assume $u$ be the inner vertex and $w$ be the outer vertex. To compute the $u$ to $w$ path we first travel to $v_m$ via $m$-edges and then reach $w$ by visiting outer edges. To compute the $w$ to $u$ path, we consider two paths (as we did in Case 1): one that goes through $v_r$ and the other that goes through $v_l$. In both scenarios, we can use the analysis of Case 1 to find a path of length at most $ \lceil {3(n-1)}/{4} \rceil + 2$.

\subsection{An Improved Upper Bound}
In this section we improve the upper bound to  $n/2+O(\sqrt{n})$ by leveraging the concept of planar separator.    Let $G $ be a    triangulated plane graph with $n$ vertices. 
 A  simple cycle $C$ in $G$ is called a \emph{simple cycle separator} if the  interior  and the exterior of $C$
 each contains at most $2n/3$ vertices.  
 Every planar graph admits a simple cycle separator of size $O(\sqrt{n})$~\cite{DBLP:journals/jcss/Miller86,DBLP:journals/acta/DjidjevV97,DBLP:conf/sigal/GazitM90}. 


Let $C$ be a simple cycle separator of size  $O(\sqrt{n})$ in $G$. Let $G_{in}$  be the graph induced by the vertices  of $C$ and the vertices inside $C$. We create a triangulated graph $G'_{in}$ by adding a vertex $s_{out}$ to all the vertices of $C$. Similarly, we define  $G_{out}$ \parthiban{to} be the graph induced by the vertices  of $C$ and the vertices outside $C$, and create   $G'_{out}$ by adding a vertex $s_{in}$ to all the  vertices of $C$. We then take a planar embedding of $G'_{out}$ such that \parthiban{$s_{in}$} lies on the outerface. Figure~\ref{figure:sep} illustrates a separator of $G$ and the corresponding $G'_{in}$ and  $G'_{out}$. 

\begin{figure}[h]
\centering
\includegraphics[width=\textwidth]{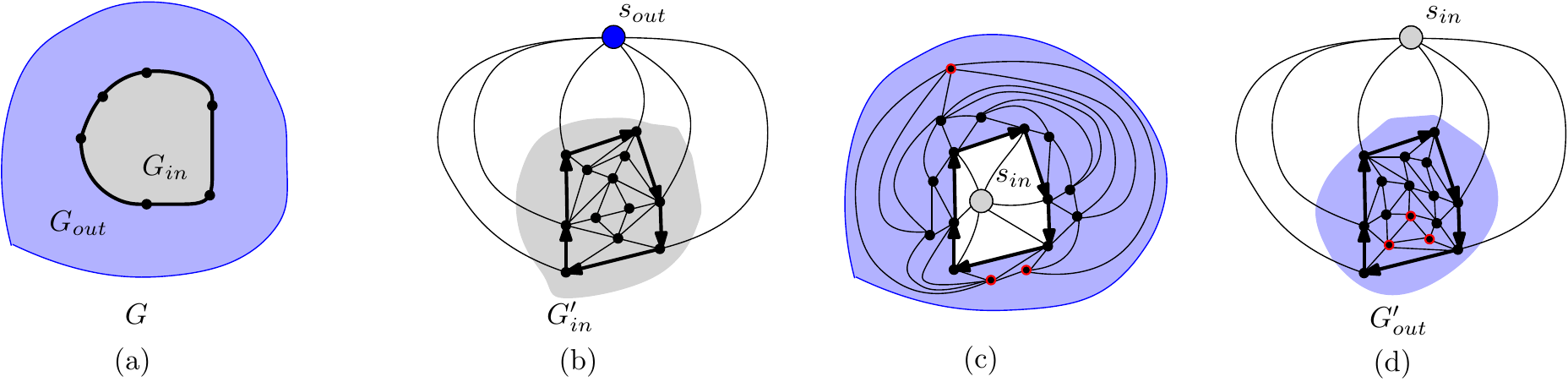}
\caption{(a) Illustration for a simple cycle separator. (b)--(d) The construction of the plane triangulations $G'_{in}$ and $G'_{out}$. The red vertices in (c) are the outer vertices and they become inner vertices in (d). } 
\label{figure:sep} 
\end{figure}

Let $n_{in}$ and $n_{out}$ be the number of vertices of $G'_{in}$ and $G'_{out}$, respectively. We now compute an orientation $\sigma$ for the edges of $G'_{in}$ such that the shortest path distance between every pair of vertices is at most $\lceil {3(n_{in}-1)}/{4} \rceil + 2$ as we did in Section~\ref{early}. We then remove the orientation of the edges on $C$ and reorient them clockwise. Let the resulting orientation be $\sigma'$. We now show that  using $\sigma'$ instead of $\sigma$ and avoiding $s_{out}$ increases the shortest path distance between a pair of vertices in $G_{in}$ by at most $O(\sqrt{n})$.  

We consider two cases depending on whether the tree $T$ rooted at $s_{out}$   has the maximum number of leaves in the Schnyder realizer of $G'_{in}$.

\bigskip
\noindent
\textbf{Case 1 ($T$ has the maximum number of leaves):} In Section~\ref{early} we constructed a path $\vv{P}$ from an inner vertex $u$ to another inner vertex $w$ such that it first moves from $u$ to an outer vertex and then traverses along the outer cycle and finally, travels towards $w$. Such a path can also be constructed using $\sigma$. However, in $\sigma'$, we need to restrict ourselves within the edges of $G_{in}$. Let $x$ be a vertex of $C$ and the first such vertex that we encounter while travelling from $v$ to $w$ following $\vv{P}$. Similarly,  let $y$ be a vertex of $C$ and the last such vertex that we encounter while travelling from $v$ to $w$ following $\vv{P}$. We now replace the subpath from $x$ to $y$ using a clockwise path on $C$. This results into a path of length at most $\lceil {3(n_{in}-1)}/{4} \rceil + O(\sqrt{n})$ that uses the orientation of $\sigma'$. We can find a $w$ to $v$ path of length at most $\lceil {3(n_{in}-1)}/{4} \rceil + O(\sqrt{n})$ using the same argument by swapping the \parthiban{roles} of $u$ and $w$. It is straightforward to observe that the argument holds even when one of $u$ and $w$, or both lie on the cycle separator. 

\bigskip
\noindent
\textbf{Case 2 ($T$ does not have the maximum number of leaves):} Without loss of generality assume that the tree rooted at an outer vertex $v\not= s_{out}$ has the maximum number of leaves in the Schnyder realizer of $G'_{in}$. The argument here is the same as that of Case 1. To reach from an inner vertex $u$ to another inner vertex $w$, we take a path $\vv{P}$ from  $u$ to $v$ and then traverse along the outer cycle and finally, travel towards $w$. Since we need to avoid $s_{out}$, we can leverage the clockwise cycle $C$. 
 Define the vertices $x$ and $y$ similar to that of Case 1. We then replace the subpath of $\vv{P}$ from $x$ to $y$ using a clockwise path on $C$. This results into a path of length at most $\lceil {3(n_{in}-1)}/{4} \rceil + O(\sqrt{n})$ that uses the orientation of $\sigma'$. We can find a $w$ to $v$ path using the same argument by swapping the role of $u$ and $w$. It is straightforward to observe that the argument holds even when one of $u$ and $w$ or both lie on the cycle separator. 

\bigskip
We now compute an orientation  $\sigma''$ for the edges of $G'_{out}$ in the same way as we did for $G'_{in}$. In the following we show that the edge orientations of $G$ obtained by taking the edge orientations from $\sigma'$  and $\sigma''$ \parthiban{ensure} an oriented diameter of $ {n}/{2}   + O(\sqrt{n})$. 

Consider a pair of vertices $u$ and $w$ in $G$. If both belong to $G_{in}$, then they can be reached from each other using a path of length $\lceil {3(n_{in}-1)}/{4} \rceil + O(\sqrt{n}) = n/2+ O(\sqrt{n})$. The same argument holds if they both belong to $G_{out}$. 

Assume now without generality that $u$ belongs to $G_{in}$ and $w$ belongs to $G_{out}$. To find a $u$ to $w$ path, we first traverse the path $P'$ that goes from $u$ to a vertex on $C$ (using $\sigma'$) and then traverse clockwise along $C$, and finally take the path $P''$ that goes from a vertex of $C$ to $w$ (using $\sigma''$). By the analysis of Section~\ref{early} and then by leveraging $\sigma'$, we obtain the length of $P'$ to be at most $\lceil(n_{in}-1)/2\rceil + O(\sqrt{n})$. By the analysis of Section~\ref{early} we know that  there are two vertex disjoint paths in $G'_{out}$ that start at $C$ and \parthiban{reach} $w$. Hence the length of one of these paths is at most  $\lceil n_{out}/2\rceil$. After considering $\sigma''$, the  corresponding path $P''$ would have a length of $\lceil n_{out}/2\rceil+O(\sqrt{n})$. Therefore, the length of the $u$ to $w$ path is upper bounded by at most $n/2+O(\sqrt{n})$. We can construct a $w$ to $u$ path by swapping the role of $u$ and $w$.

The computation of the simple cycle separator takes $O(n)$ time~\cite{DBLP:journals/jcss/Miller86}. The computation of the required Schneider realizer also takes $O(n)$ time~\cite{ZhangH05}. Hence it is straightforward to implement the algorithm in linear time. The following theorem summarizes the result of this section.
\begin{theorem}\label{th:od}
The oriented diameter of a planar triangulation with $n$ vertices is $n/2+O(\sqrt{n})$ and such an orientation can be computed in $O(n)$ time.
\end{theorem}

\section{Lower Bound}
The lower bound is determined by the nested triangles graph $G_n$, which is defined as follows. 

For $n=3$, the graph $G_3$ is a cycle $u_1,v_1,w_1$ of three vertices. For $n=3m$, where $m>1$ is a positive integer, $G_n$ is obtained by enclosing $G_{n-1}$ inside a cycle $u_{n/3},v_{n/3},w_{n/3}$ of three vertices and then adding the edges $(u_{n/3-1},u_{n/3})$, $(v_{n/3-1},v_{n/3})$, and $(w_{n/3-1},w_{n/3})$. See Figure~\ref{figure:nested}(a).

\begin{figure}[h]
\centering
\includegraphics[width=\textwidth]{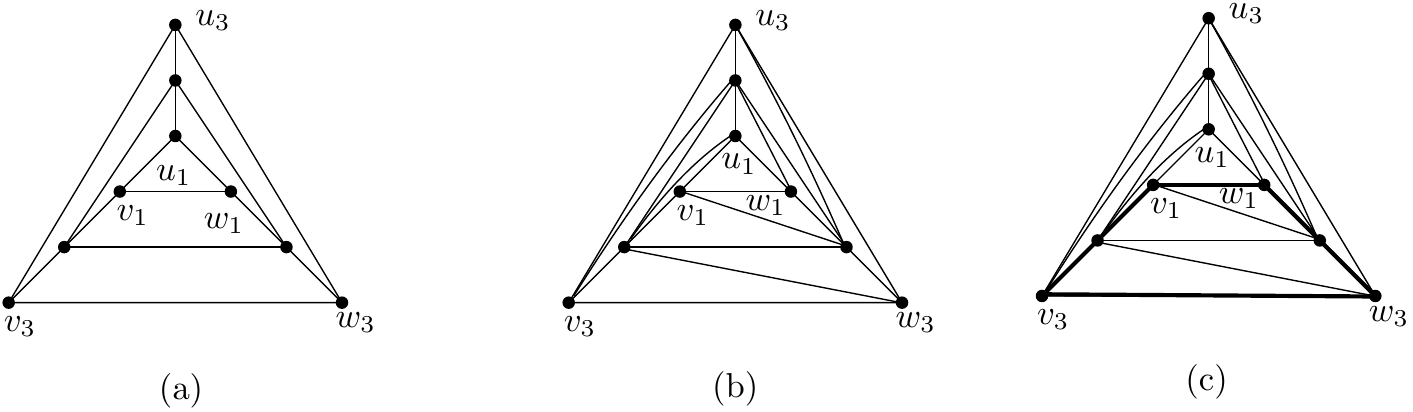}
\caption{(a) Illustration for a nested triangles graph. (b) A triangulated nested triangles graph. (c) A shortest cycle of $2n/3$ vertices through $w_1$ and $v_{n/3}$.}
\label{figure:nested}
\end{figure}

We triangulate the graph $G_n$ by adding the edges $(v_i,w_{i+1}), (v_i,u_{i+1}), (w_i,u_{i+1})$, where $1\le i <n/3$. See Figure~\ref{figure:nested}(b). It is now straightforward to employ an induction on $n$ to show  that the shortest path length between $w_1$ and $v_{n/3}$ is $n/3$. Therefore, any orientation of the graph, a directed cycle through  $w_1$ and $v_{n/3}$ must be of length at least $2n/3$. Hence we obtain the following theorem.

\begin{theorem}
For every $n=3m$, where $m$ is a positive integer, there exists a planar triangulation with oriented diameter at least $n/3$.
\end{theorem}

 

\section{Planar Weighted Oriented Diameter}

In this section we consider the weighted oriented diameter. Formally, given an edge weighted graph, the \emph{weighted oriented diameter} is the minimum weighted diameter over all of its strongly connected orientations.

If the edge weights are small, then one can find an orientation with small weighted diameter. The following corollary is a direct consequence of Theorem~\ref{th:od}. 

\begin{corollary}
Let $G$ be an edge weighted planar graph with $n$ vertices. Assume that the weight of each edge is at most ${(1+\epsilon)}/{n}$, where $0\le \epsilon<1$, and the sum of all weights in $1$. Then in linear time, one can compute an orientation of $G$ with weighted oriented diameter at most $\frac{(1+\epsilon)}{2} + \frac{1}{O(\sqrt{n})}$.
\end{corollary}

In the following section we show that the decision version of the weighted oriented diameter problem is weakly NP-complete.}

\jyoti{\subsection{Planar Weighted Oriented Diameter is Weakly NP-complete}
Here we show that given a planar  graph $G$ and an integer $k$, it is weakly NP-complete to decide whether $G$ has a strongly connected orientation of diameter at most $k$. We formally define the problem as follows:

\begin{enumerate}
\item[] \textbf{Problem:} \textsc{Planar Weighted Oriented Diameter}
    \item[] \textbf{Input:} An undirected connected planar graph $G$, where each edge is weighted with a positive real number, and  a positive integer $D$.
    \item[] \textbf{Question:} Is there a strongly connected orientation of $G$  such that the weighted diameter of the resulting oriented graph is at most $D$?
\end{enumerate}
 
The problem is in NP because given an orientation of the edges of $G$, one can verify whether the weighted oriented diameter is within $D$ in polynomial-time by computing all pair shortest paths.

We will reduce the weakly NP-complete problem \textsc{Partition}, which is defined as follows:

\begin{enumerate}
\item[] \textbf{Problem:} \textsc{Partition}
    \item[] \textbf{Input:} A  multiset $S$ of positive integers.
    \item[] \textbf{Question:} Can $S$ can be partitioned into two subsets with equal sum?
\end{enumerate}

Given an instance $S = \{a_1,a_2,\ldots,a_n\}$ of \textsc{Partition}, we construct an instance $I = (G,D)$ of \textsc{Planar Weighted Oriented Diameter} as follows:

\begin{enumerate}
\item[] \textit{Step 1.} Take a $2\times (n+1)$ grid graph, e.g.,  Figure~\ref{figure:hard}(a). Let $v_1,\ldots,v_{n+1}$ be the vertices at the top row and let $w_1,\ldots,w_{n+1}$ be the vertices at the bottom row. and set the weights of  the top edges $(v_1,v_2),\ldots, (v_n,v_{n+1})$ using the numbers $a_1,\ldots,a_n$ from left to right.

\item[] \textit{Step 2.} Subdivide the leftmost vertical edge $(v_1,w_1)$ with a division vertex $s$. Similarly, subdivide the rightmost vertical edge with a division vertex $t$. 

\item[] \textit{Step 3.} Replace each internal vertical edge $(v_i,w_i)$, where $1<i<n+1$, with a cycle $v_i,d_{v_iw_i},w_i,d'_{v_iw_i},v_i$, where $d_{v_iw_i}$ and $d'_{v_iw_i}$ are two new vertices, e.g.,  
Figure~\ref{figure:hard}(b). We will refer to these cycles as \emph{inner cycles}.

\item[] \textit{Step 4.} Set the weight of all the edges except for the top edges to $\epsilon = \frac{1}{m}$, where $m$ is the number of edges in the graph. Set $D$ to be equal to $1+\frac{1}{2}\sum_{i=1}^{n}a_i-\epsilon$.

\end{enumerate}
\begin{figure}[h]
\centering
\includegraphics[width=\textwidth]{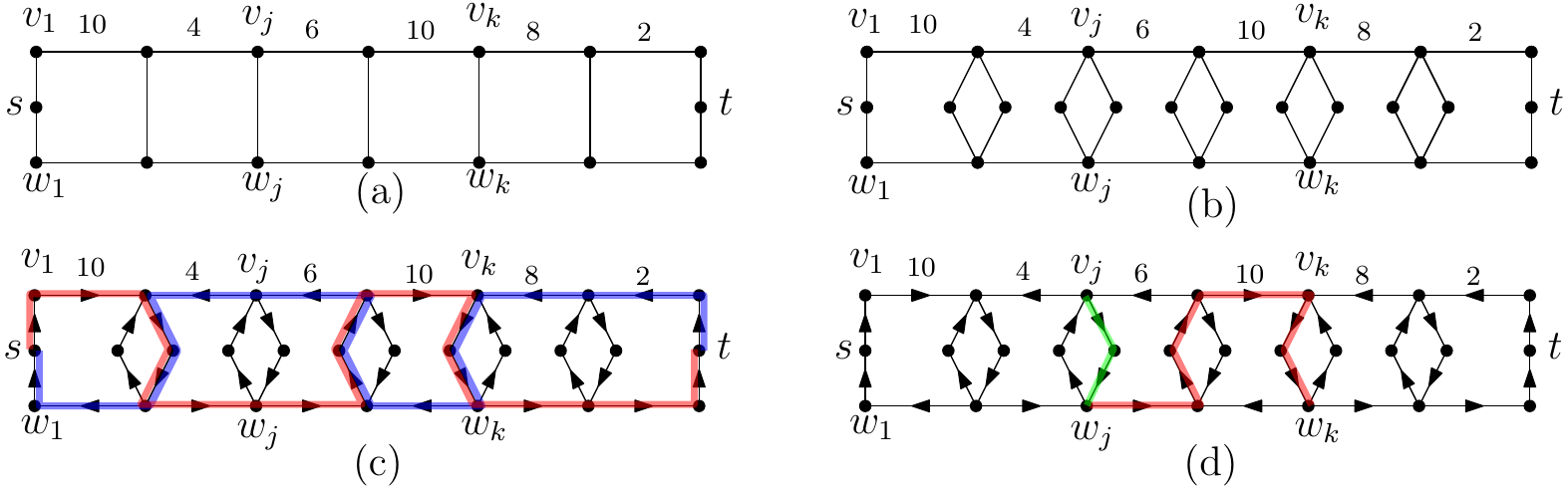}
\caption{(a)--(b) Construction of $G$ from $S$. (c)--(d) Illustration for the orientation.} 
\label{figure:hard} 
\end{figure}

\begin{lemma}
\label{cross}
Let $\vv{G}$ be a strong orientation of $G$. Let  $\vv{P}$ and  $\vv{Q}$ be the directed paths from $s$ to $t$ and from $t$ to $s$, respectively. Then $\vv{P}$ and  $\vv{Q}$ partition the top edges of $G$. 
\end{lemma}
\begin{proof}
We first show that every top edge either   belongs  to $\vv{P}$ or to $\vv{Q}$. Suppose for a contradiction that there exists a top edge $(v_i,v_{i+1})$  that neither belongs  to $\vv{P}$ nor to $\vv{Q}$. Then $(w_i,w_{i+1})$ must be a cut edge in $G$, i.e., removal of  $(w_i,w_{i+1})$ would destroy any path that connects $s$ to $t$ and  $t$ to $s$. Therefore, this edge must belong to both $\vv{P}$ and  $\vv{Q}$. Assume without loss of generality that $(w_i,w_{i+1})$ is directed from $w_i$ to $w_{i+1}$, then no directed path from $t$ can reach $w_i$, which contradicts that $(w_i,w_{i+1})$ belongs to $\vv{Q}$.

We now show that no top edge can belong to both  $\vv{P}$ and  $\vv{Q}$. Suppose for a contradiction that a top edge $(v_i,v_{i+1})$    belongs  to both $\vv{P}$ and $\vv{Q}$.  Assume without loss of generality that $(v_i,v_{i+1})$ is directed from $v_i$ to $v_{i+1}$.  To reach   $v_i$, the path $\vv{Q}$ must  first traverse $w_{i+1}$ to $w_i$. However, in this case, the path cannot be extended to reach both $s$ and $v_{i+1}$.
\end{proof}

In the following we show that $S$ admits a bipartition if and only if $G$ admits a strong orientation with  diameter   at most $D$.}

\jyoti{\subsubsection{From Bipartition to Oriented Diameter}
 
Assume   that $S$ admits a partition into two sets $A$ and $B$  such that sum of the elements in $A$ is equal to that of $B$. We now show how to obtain a strong orientation for $G$ such that the diameter is at most $D$. 

\bigskip\noindent
\textbf{Orientation of $G$:} We first orient the top edges corresponding to $A$ from left to right. We then construct a directed $s$ to $t$ path $\vv{P}$ by using these edges but avoiding the top edges that correspond to $B$ as follows. We start at $s$ and for each subsequent top edge $(v_i,v_{i+1})$ of $A$, we connect the current vertex to $v_i$ either using a vertical edge (i.e., if $v_i=v_1$), or using a path that first moves downwards to  the bottom row, then traverses to the right and finally moves upward to reach $v_i$. Let $v_k$ be the last vertex on the top row that corresponds to $A$. Then $v_k$ is connected to $t$ either using a vertical edge (i.e., if $v_{k}=v_{n+1}$), or using a path that  first moves downwards to  the bottom row, then traverses to the right and finally moves upward to reach $t$.  Figure~\ref{figure:hard}(c) illustrates such a path in red where $A=\{10,10\}$. }

We now show that there is an undirected $s$ to $t$ path that contains each top edge  corresponding to $B$. By the construction of $\vv{P}$, we can observe that for each $j$ from 2 to $n$, $v_j$ and $w_j$ are still connected by an undirected path. Therefore, to restrict  $s$ from reaching the first top edge of $B$ using an undirected path, $\vv{P}$ needs to include both $(v_i,v_{i+1})$ and $(w_i,w_{i+1})$ where $i<j$. However,  $\vv{P}$ never includes both $(v_i,v_{i+1})$ and $(w_i,w_{i+1})$. In the same way we can show that the undirected path can be extended to include all the top edges of $B$ and to reach $t$. We then make the path a directed path $\vv{Q}$ from $t$ to $s$.   Figure~\ref{figure:hard}(c) illustrates such a path in blue where $B=\{4,6,2,8\}$. 

Finally, we orient the edges of the inner cycles to form a directed cycle as follows. Let $v_i,d_{v_iw_i},w_i,d'_{v_iw_i},v_i$ be an inner cycle. If at most one of the paths among $v_i,d_{v_iw_i},w_i$ and $v_i,d'_{v_iw_i},w_i$ are already oriented, then we can orient the other path to form a directed cycle. If both these paths are oriented and if they are oriented in opposite directions (one upward and the other downward), then we already have a directed cycle. Otherwise, the paths are oriented in the same direction. In this case, we  reverse the orientation of one of these paths to obtain a directed cycle, which does not change the lengths of $\vv{P}$ and  $\vv{Q}$.   

\bigskip\noindent
\textbf{Bound on Weighted Diameter:} By construction, each of  $\vv{P}$ and $\vv{Q}$ is of length at most $\frac{n}{m}+\frac{1}{2}\sum_{i=1}^{n}a_i<D$.

Consider now a pair of vertices $v_j,v_k$ on the top row, where $j>k$. We now show how to find a $v_j$ to $v_k$ path of length at most $D$. If they both lie on $\vv{P}$ (or on $\vv{Q}$), then they are already connected by a path of length at most $D$. Otherwise, without loss of generality assume that $v_j$ lies on $Q$ and $v_k$ does not lie  on $Q$. By Lemma~\ref{cross}, $v_k$ must lie on $P$. Since  $\vv{P}$ starts at $s$ and to reach $v_k$, it must pass through $w_j$. Therefore, we can move downward from $v_j$ to reach  $w_j$ and then follow $\vv{P}$ to reach $v_k$, e.g.,  Figure~\ref{figure:hard}(d). Thus the length of the path would be at most $\frac{n}{m}+\frac{1}{2}\sum_{i=1}^{n}a_i<D$.

Consider now a pair of vertices $w_j,w_k$ on the bottom row. Since $n<m/2$ and since $v_j$ and $v_k$ are connected to $w_j$ and $w_k$ with the directed inner cycles, respectively,  $w_j$ and  $w_k$ are  connected to each other with paths of length at most $D$ (through $v_j$ and $v_k$).    A similar argument applies to every pair of inner vertices of $G$, and also to each  pair that contain an inner vertex and an outer vertex.

\jyoti{\subsubsection{From Oriented Diameter to Bipartition}

Assume that $G$ admits a strong orientation $\vv{G}$ with  diameter is at most $D$. We now show how to find the required partition for $S$. Let $\vv{\sigma_{st}}$ be the directed shortest path from $s$ to $t$. Similarly, let  $\vv{\sigma_{ts}}$ be the directed shortest path from $t$ to $s$. By Lemma~\ref{cross}, $\vv{\sigma_{st}}$ and  $\vv{\sigma_{ts}}$ partition the top edges of $G$. Let the corresponding subsets of $S$ be $A$ and $B$, respectively. If the sum of the numbers in $A$ is strictly smaller than that of $B$, then the sum of the numbers of $B$ is at least $1+\frac{1}{2}\sum_{i=1}^{n}a_i > D$. This contradicts that the diameter of $\vv{G}$ is at most $D$.   The same argument holds when then the sum of the numbers of $B$ is strictly smaller than that of $A$. Therefore, $A$ and $B$ must be the required partition. 

\smallskip\noindent
The following theorem summarizes the result of this section.
\begin{theorem}
\textsc{Planar Weighted Oriented Diameter} is weakly NP-complete.
\end{theorem}
Since the pathwidth of $G$ is $O(1)$, we obtain the following corollary. 
\begin{corollary}
\textsc{Planar Weighted Oriented Diameter} is weakly NP-complete, even when the pathwidth of the input graph is bounded by a constant.
\end{corollary}
}

\section{Conclusion}

 In this paper we computed exact value of the oriented diameter for triangular grid graphs and proved an $n/3$ lower  bound and an  $n/2+O(\sqrt{n})$ upper
 bound on the oriented diameter of planar triangulations. A natural direction for future research would be to close the gap between the lower bound and the upper bound. 
 We also showed that the weighted version of the oriented diameter problem is weakly NP-complete.  
 Although the time complexity of  the unweighted version  remains open, it would be interesting to examine whether the weighted version is strongly NP-hard.

\bibliographystyle{abbrv}
\bibliography{bib}

\end{document}